\documentclass[a4paper,UKenglish,cleveref, autoref, thm-restate]{lipics-v2021}

\usepackage{algorithm} 

\newlength{\Oldarrayrulewidth}
\newcommand{\Cline}[2]{%
  \noalign{\global\setlength{\Oldarrayrulewidth}{\arrayrulewidth}}%
  \noalign{\global\setlength{\arrayrulewidth}{#1}}\cline{#2}%
  \noalign{\global\setlength{\arrayrulewidth}{\Oldarrayrulewidth}}}

  \newcolumntype{?}{!{\vrule width 1.25pt}}

  \usepackage{algorithmicx}

\pdfoutput=1 
\hideLIPIcs  

\title{Winning Without Observing Payoffs: Exploiting Behavioral Biases to Win Nearly Every Round} 

\titlerunning{Winning Without Observing Payoffs}

\author{Avrim Blum}{Toyota Technological Institute at Chicago, IL, USA}{avrim@ttic.edu}{https://orcid.org/0000-0003-2450-5102}{The author was supported in part by the National Science Foundation under grant CCF-2212968, and by ECCS-2216899.}

\author{Melissa Dutz}{Toyota Technological Institute at Chicago, IL, USA}{melissa@ttic.edu}{https://orcid.org/0009-0008-0971-8285}{The author was supported in part by an NSF CSGrad4US Fellowship.}

\authorrunning{A. Blum and M. Dutz} 

\Copyright{Avrim Blum and Melissa Dutz} 



\relatedversion{} 

\funding{}

\nolinenumbers 

\EventEditors{Venkatesan Guruswami}
\EventNoEds{1}
\EventLongTitle{15th Innovations in Theoretical Computer Science Conference (ITCS 2024)}
\EventShortTitle{ITCS 2024}
\EventAcronym{ITCS}
\EventYear{2024}
\EventDate{January 30 to February 2, 2024}
\EventLocation{Berkeley, CA, USA}
\EventLogo{}
\SeriesVolume{287}
\ArticleNo{18}

\begin{document}

\maketitle

\begin{abstract}
Gameplay under various forms of uncertainty has been widely studied. Feldman et al. \cite{feldman2010playing} studied a particularly low-information setting in which one observes the opponent's actions but no payoffs, not even one's own, and introduced an algorithm which guarantees one's payoff nonetheless approaches the minimax optimal value (i.e., zero) in a symmetric zero-sum game. Against an opponent playing a minimax-optimal strategy, approaching the value of the game is the best one can hope to guarantee. However, a wealth of research in behavioral economics shows that people often do not make perfectly rational, optimal decisions. Here we consider whether it is possible to actually win in this setting if the opponent is behaviorally biased. We model several deterministic, biased opponents and show that even without knowing the game matrix in advance or observing any payoffs, it is possible to take advantage of each bias in order to win nearly every round (so long as the game has the property that each action beats and is beaten by at least one other action). We also provide a partial characterization of the kinds of biased strategies that can be exploited to win nearly every round, and provide algorithms for beating some kinds of biased strategies even when we don't know which strategy the opponent uses. 
\end{abstract}

\section{Introduction}
Game theory has traditionally assumed all agents are perfectly rational utility maximizers \cite{Osborne1994}, but behavioral economics shows this assumption is often unsound; people often deviate from ``optimal'' behavior in predictable ways, exhibiting behavioral biases such as loss aversion (exhibiting a higher sensitivity to losses than equivalent gains) and confirmation bias (interpreting information in a way which favors existing beliefs) \cite{kahneman2011thinking, cartwright2018behavioral}. These insights initiated the study of behavioral game theory, which has demonstrated that people employ consistent ``irrational'' strategies in wide-ranging strategic settings (see \cite{2011behavioralgttext}). Given the evidence that people often use behavioral strategies in games, it is important to understand the implications of this in algorithmic interactions (see, e.g., \cite{halpernalgrationality, AGARWAL2023119117, securitygames, learningexploitationbias}).

In this work, we are interested in the extent to which behavioral biases are vulnerable to exploitation in competitive games. To study this, we consider symmetric, repeated, two-player, zero-sum games, and find that even an uninformed opponent who does not know the game or observe any payoffs can take advantage of a wide variety of biased strategies to beat the biased player in nearly every round.  

The setting we consider here is based on work by Feldman et al. \cite{feldman2010playing}. Gameplay under various forms of uncertainty is a focus of much research interest; for example, the bandit setting, in which a player observes payoffs after each round but not the opponent's play, has been widely studied (see \cite{auer2002nonstochastic, predlearninggamesbook}). Feldman et al. introduced a complementary setting, in which the player observes the opponent's play but no payoffs. No assumptions are made about the opponent, who may have full knowledge of the game. In particular, Feldman et al. consider symmetric, two-player, repeated games. They introduce a ``copycat'' algorithm which guarantees a payoff approaching the value of the game, which is the best bound one could hope to achieve against an opponent playing a minimax-optimal strategy. For a symmetric zero-sum game, the value of the game is zero, so their algorithm guarantees a payoff approaching zero in that case. The Copycat algorithm basically works by ensuring that for any pair of actions $a_1$ and $a_2$, the number of rounds in which the uninformed player plays $a_1$ and the opponent plays $a_2$ is not very far from the number of rounds in which the player plays $a_2$ and the opponent plays $a_1$. In this way, the uninformed player is able to achieve roughly the value of the game without knowing (or learning) anything about payoffs. 

The setting introduced by Feldman et al. \cite{feldman2010playing} is a particularly interesting setting in which to study behavioral bias, because it allows us to consider to what extent one can exploit an opponent's biases when the opponent's behavior is essentially the only information one has. We consider the same basic setting, with a few additional restrictions. Like Feldman et al., we consider repeated, two-player, zero-sum, symmetric games in a setting where one does not have a priori knowledge of the game matrix and does not observe payoffs, but does observe the opponents actions. We also assume the opponent is deterministic and plays according to some behaviorally-biased strategy. For simplicity, we consider games where payoffs are limited to \{1, 0, -1\} (just wins, ties and losses). We assume no action is unbeatable and each action beats at least one other action.

There are two technical components to beating a biased opponent in this setting: one needs to learn to predict the opponent's actions, and one needs to learn a best response to those actions the opponent plays. We note that it is not sufficient to just be able to predict accurately. For example, one natural idea which does not work (even if one is willing to take exponential time) is to simply find the lexicographically first game matrix consistent with all the opponent's actions so far and their behavioral strategy, and then play a best response to the action predicted by that game matrix. It is possible that two game matrices, $M1$ and $M2$, could be consistent with the opponent's play, but playing best responses with respect to $M1$ could result in losing most  rounds (and tying in the remaining rounds) if the true game matrix is $M2$. We show such an example in section \ref{predicting_insufficient}.

We model several deterministic, behaviorally-biased opponents, and show how to exploit each bias to win nearly every round. We also give a partial characterization of the kinds of behavioral strategies we can exploit in general, and show that in some cases, we can also win nearly every round against a biased opponent even if we do not know which behavioral strategy they use.

This work highlights the danger of using behaviorally-biased strategies in competitive settings. We might expect that, while such strategies are not minimax optimal, some seemingly reasonable ones might do well enough, particularly against an uninformed opponent. However, our work shows that some behavioral biases can be exploited even by a player without prior knowledge of the game matrix entries who does not observe payoffs. 

\section{Setting} \label{setting}
We consider two-player, repeated games with the following properties: 

\begin{itemize}
    \item Symmetric
    \item Zero-sum
    \item Payoffs are limited to $\{1,0,-1\}$ (representing win, tie, and loss respectively)
    \item Each action beats at least one other action
    \item Each action loses to at least one other action 
\end{itemize}

\begin{definition}[Permissible Game] \label{def:permissible_game}
We refer to a game with the above properties as a permissible game for convenience.
\end{definition}

Note that a symmetric game with $n$ actions ($a_1$, ..., $a_n$) can be represented by an $n \times n$ payoff matrix $M$, where $M[i,j]$ is the payoff for the row player and $-M[i,j] = M[j,i]$ is the payoff of the column player when the former plays $a_i$ and the latter plays $a_j$. 

We consider how well we can do without knowing the payoff matrix $M$ or even observing payoffs when playing against a behaviorally-biased opponent who does know $M$. Our goal will be to win nearly every round.

\subsection{Models of behaviorally-biased opponents}
We assume all opponents know the payoff matrix and play a deterministic, behaviorally-biased strategy. The following are the particular biased adversaries we consider here: 

    \paragraph*{Myopic Best Responder} The opponent plays a best response to our last action. If there are multiple best responses to our last action, we assume the opponent breaks ties according to some fixed but unknown order over actions. On the first round, since there is no previous action, we assume the opponent plays the first action in that ordering. 
    
    This is our simplest opponent, who just assumes that tomorrow will be exactly like today.
    
    \paragraph*{Gambler's Fallacy} The opponent plays a best response to our ``most overdue'' (least-frequently-played) action. This opponent is motivated by the Gambler's Fallacy, in which one assumes an event which hasn't happened as often as expected is ``overdue'' and is more likely to occur (for example, when tossing a fair coin, the belief that Tails is more likely following a series of Heads). 
    
    If there are multiple best responses to our ``most overdue'' action, or if there are multiple ``most overdue'' actions the opponent might try to counter (for example, at the start of the game when all actions are equally ``overdue'' since they have all been played zero times), the opponent breaks ties using a fixed but unknown order over actions.
    
    \paragraph*{Win-Stay, Lose-Shift} The opponent continues to play the same action following a win, and switches to a new action following a loss. We consider two variants of this strategy: one in which the opponent behaves as if a tie is a win, continuing to play the same action after tying, and one in which the opponent behaves as if a tie is a loss, switching actions after tying. 

    We assume that when the opponent shifts actions, they shift to the action immediately following the current action in some fixed but unknown order over actions. In the first round, we assume the opponent plays the first action in that ordering.
    
    Win-stay lose-shift was initially proposed as an early algorithm for bandit problems by Robbins \cite{bams/1183517370}. There is also some empirical research indicating that people may play rock-paper-scissors according to an approximate version of this strategy: Wang et al. \cite{Wang_2014} observed that players are more likely to play the same action again after winning, and more likely to shift to the next action in the sequence rock-paper-scissors following a loss. More broadly, win-stay lose-shift has been proposed as a potential explanation for the evolution of cooperative behavior \cite{doiStrategyWinstay}. 

    This strategy can also be viewed as similar to the hot-hand fallacy, in which one believes an action which has just won is having a ``hot streak'' and is more likely to win in the future, and conversely believes an action which has just lost is having a ``cold'' streak and is more likely to lose in the future. 
        
    \paragraph*{Follow the Leader} The opponent plays the best action in retrospect, where the best action in retrospect is the action that would have achieved the highest net payoff against our history of play. If multiple actions are tied for the highest payoff, the opponent breaks ties using some fixed but unknown order over actions. In the first round, when there is no history to choose based on, the opponent plays the first action in that ordering. 
    
    It has been hypothesized that irrational strategic behavior may be explained by agents playing a misspecified or simplified version of a game rationally \cite{DEVETAG2008364}. The above strategy is reasonable for an opponent who erroneously models our play as stochastic. In online learning, this strategy is called Follow-the-Leader or FTL (see \cite{shalev2012online}). 
    
    We also consider a limited-history variant of this opponent, in which the the opponent only considers the last $r$ rounds, where $r > 0$.

    \paragraph*{Highest Average Payoff} The opponent plays the action that has achieved the highest average payoff over the times they have played it. We assume the opponent initializes the average payoff of each action as 0. If multiple actions are tied for the highest average payoff, the opponent breaks ties using some fixed but unknown order over actions.

    This is a reasonable strategy for an opponent who believes our play is stochastic, in which case its actions can be viewed as coins with various biases. It also corresponds to a well-known reinforcement learning algorithm, epsilon-greedy Q-learning, with $\epsilon=0$ (see \cite{sutton2018reinforcement}).

    \section{Preliminaries and Intuition}

While the strategy of playing action $i$ against action $j$ roughly as many times as we play action $j$ against action $i$ (for all $i$, $j$) employed by the Copycat algorithm in Feldman et al. \cite{feldman2010playing} is sufficient to roughly tie in a symmetric zero-sum game, we'd like to go further and win nearly every round. To do this, we will learn to predict the opponent's actions and learn best responses to those actions the opponent plays. It turns out that predicting the opponent's actions is the ``easy'' part, and the main challenge will be learning best responses. We discuss each of these at a high level below.

\paragraph*{Predicting the Opponent's Actions} \label{sec_general_prediction}
It is possible to learn to predict any of our opponents' actions using a simple halving algorithm (see \cite{littlestone1988learning}). Recall that in our setting, an opponent's choice of action is deterministic and induced by their known biased strategy, their tie-breaking mechanism, the unknown game matrix, and the history of play. Consider all possible combinations of the opponent's known biased strategy with a potential tie-breaking mechanism to be a family of deterministic strategies. Then, consider all possible pairs including a deterministic strategy from that family and a game matrix of the appropriate size. Predict the opponent's next action by taking the majority vote over the predictions corresponding to these pairs and the play history. Then, eliminate all pairs that did not correspond to a correct prediction. Each time a mistake is made, at least half of the remaining pairs must have predicted incorrectly, so at least half of the remaining pairs are eliminated. Given a finite number of possible tie-breaking mechanisms, this algorithm clearly results in a finite number of prediction errors. However, it is inefficient: even with our assumption of only 3 possible payoff values, there are roughly $3^{n^2}$ possible game matrices, and for each game matrix, there may be many tie-breaking mechanisms to consider: throughout this paper we generally assume the opponent breaks ties according to some fixed ordering over actions, in which case  there are $n!$. We will give more efficient prediction algorithms for each of the opponents we introduced above by taking further advantage of their specific strategies. 

\begin{algorithm}[H] 
 \caption{Generic Algorithm for Predicting the Opponent's Next Action \label{alg_prediction_by_halving}} 
    \hspace*{\algorithmicindent} \textbf{Input:}  A family of deterministic strategies $\mathbb{S}$ and the size of the game matrix $n$  \\
    \hspace*{\algorithmicindent} \textbf{Output:} The action we predict the opponent will play next \\
    \begin{enumerate}
        \item initialize $hist$, the play history, to the empty list [].
        \item $\mathbb{H} \gets$ all possible pairs $(S,G)$ including a strategy $S \in \mathbb{S}$ and a game matrix $G$ of size $n$.
        \item Predict the opponent's action as the majority vote among $S(G, hist) \quad \forall (S, G) \in \mathbb{H}$, breaking ties arbitrarily. 
        \item Play some action $a_{our}$ and observe the opponent's actual action, $a_{opponent}$.
        \item Remove from $\mathbb{H}$ all pairs which did not correspond to correct predictions ($\mathbb{H} = \mathbb{H} \setminus (S,G) : S(G, hist) \neq a_{opponent}$). 
        \item Append $(a_{our}, a_{opponent})$ to $hist$.
        \item Return to step 3. 
    \end{enumerate}
\end{algorithm}

\begin{theorem} \label{thm_halving_bounds} 
   Given an opponent playing according to a consistent deterministic strategy which breaks ties among actions according to fixed ordering over actions, Algorithm \ref{alg_prediction_by_halving} makes up to O($n^2$) prediction mistakes in $O(3^{n^2}n!R)$ time per round, where $R$ is the runtime in which the opponent's strategy selects an action given the game matrix and history. 

   In general, if the opponent plays according to a consistent deterministic strategy in $\mathbb{S}$, Algorithm \ref{alg_prediction_by_halving} makes up to O$(n^2 + log(|\mathbb{S}|))$ prediction mistakes in $O(3^{n^2}|\mathbb{S}|R)$ time per round, where $O(R)$ is the time in which any strategy in $\mathbb{S}$ returns an action given the game matrix and history. 
\end{theorem}

\begin{proof}
    At each round, we consider the predictions given by $|\mathbb{S}|$ possible strategies paired with up to $O(3^{n^2})$ possible game matrices. If any strategy in $\mathbb{S}$ chooses an action in $O(R)$ time, this takes $O(3^{n^2}R)$ time per round. 
    
    To predict, we begin with $3^{n^2}|\mathbb{S}|$ hypotheses (pairs of deterministic strategies and possible game matrices). We assume the opponent plays one of the deterministic strategies in $\mathbb{S}$, so their strategy paired with the correct game matrix of size $n$ will always predict correctly. We consider all possible game matrices of size $n$, so this pair must be in our initial $\mathbb{H}$. Since it never makes a mistake, it will never be eliminated. Therefore when the size of $\mathbb{H}$ is reduced to 1 we are guaranteed to always predict correctly as there is a perfect hypothesis in $\mathbb{H}$. Because we predict according to the majority vote among hypotheses in $\mathbb{H}$, each time we make a mistake, at least half of the pairs correspond to an incorrect prediction and are eliminated. We can halve the size of $\mathbb{H}$ at most $\log(|\mathbb{H}|) = log(3^{n^2}|\mathbb{S}|)$ times, so we can make at most that many prediction mistakes.
\end{proof}

\paragraph*{Learning Best Responses}
It may seem like being able to predict generally is sufficient to learn best responses, but this is not the case. For example, a general (exponential) algorithm which does not work is to find the lexicographically first game matrix and action ordering pair consistent with all the opponent’s actions so far and their behavioral strategy, and then play a best response to the action predicted by that pair according to the corresponding game matrix. However, even for the simplest opponent we consider, the Myopic Best Responder, and identical orderings over actions, there could be two game matrices $M$ and $M^*$ which are consistent with the opponent’s play, but playing best responses with respect to $M$ could result in losing most rounds (and tying in the remaining rounds) if the true game matrix is $M^*$ (see the example below). Instead, we will more actively exploit the opponent's particular strategy to elicit best responses.

\paragraph*{Example: Playing Optimally According to a Consistent Game Matrix and Action Ordering is Not Sufficient to Win} \label{predicting_insufficient}

   \begin{table}[H]
    \setlength{\extrarowheight}{2pt}
    \begin{tabular}{cc ? c|c|c?c|c|c?}
      & \multicolumn{1}{c}{} & \multicolumn{6}{c}{$M^*$} \\
      & \multicolumn{1}{c}{} & \multicolumn{1}{c}{$R$} & \multicolumn{1}{c}{$P$} & \multicolumn{1}{c}{$S$} & \multicolumn{1}{c}{$R'$}  & \multicolumn{1}{c}{$P'$} & \multicolumn{1}{c}{$S'$} \\\Cline{1.25pt}{3-8}
      \multirow{6}* & $R$ & $0$ & $-1$ & 1 & 1 & 0 & \textbf{0} \\\cline{3-8} 
      & $P$ & 1 & 0 & -1 & \textbf{1} & 1 & 0 \\\cline{3-8}
      & $S$ & -1 & 1 & 0 & 0 & \textbf{1} & 1 \\\Cline{1.25pt}{3-8}
      & $R'$ & -1 & \textbf{-1} & 0 & 0 & -1 & 1 \\\cline{3-8}
      & $P'$ & 0 & -1 & \textbf{-1} &  1 & 0 & -1 \\\cline{3-8}
      & $S'$ & \textbf{0} & 0 & -1 & -1 & 1 & 0 \\\Cline{1.25pt}{3-8}
    \end{tabular} \ \ \ \ \ \ \ \ \ 
    \begin{tabular}{cc ? c|c|c?c|c|c?}
      & \multicolumn{1}{c}{} & \multicolumn{6}{c}{$M$} \\
      & \multicolumn{1}{c}{} & \multicolumn{1}{c}{$R$} & \multicolumn{1}{c}{$P$} & \multicolumn{1}{c}{$S$} & \multicolumn{1}{c}{$R'$}  & \multicolumn{1}{c}{$P'$} & \multicolumn{1}{c}{$S'$} \\\Cline{1.25pt}{3-8}
      \multirow{6}* & $R$ & $0$ & $-1$ & 1 & 1 & 0 & \textbf{-1} \\\cline{3-8} 
      & $P$ & 1 & 0 & -1 & \textbf{-1} & 1 & 0 \\\cline{3-8}
      & $S$ & -1 & 1 & 0 & 0 & \textbf{-1} & 1 \\\Cline{1.25pt}{3-8}
      & $R'$ & -1 & \textbf{1} & 0 & 0 & -1 & 1 \\\cline{3-8}
      & $P'$ & 0 & -1 & \textbf{1} &  1 & 0 & -1 \\\cline{3-8}
      & $S'$ & \textbf{1} & 0 & -1 & -1 & 1 & 0 \\\Cline{1.25pt}{3-8}
    \end{tabular}
  \end{table} 

    Action Ordering $\Omega$: R, P, S, R', P', S'\\
    
Suppose we are playing a game $M^*$ against the Myopic Best Responder, whose ordering over actions is $\Omega$. Suppose that we predict the Myopic Best Responder's actions according to the correct action ordering $\Omega$ but a slightly different game matrix $M$, and play best responses according to $M$. Both games are like an expanded version of rock-paper-scissors, in which the interactions among R, P, S and among R',P', S' respectively are the same as in standard rock-paper-scissors, but the interactions between those two sets of actions differs between $M$ and $M^*$; the differing entries are bolded for convenience. Consider the following sequence of play, which is consistent with $M^*$, $M$, $\Omega$, and the strategy of the Myopic Best Responder, and optimal according to $M$:

\begin{table}[H]
\begin{tabular}{|l|c|c|c|c|c|c|c|}
\hline
Round & 1 & 2 & 3 & 4 & 5 & 6 & ...\\
\hline
MBR's Action & R & S & P & R & S & P & ...\\
Our Action & S' & P' & R' & S' &  P' & R' & ...\\
Anticipated Payoff (According to M) & 1 & 1 & 1 & 1 & 1 & 1 & ... \\
Actual Payoff (According to M*) & 0 & -1 & -1 & 0 & -1 & -1 & ...\\
\hline
\end{tabular}
\caption{Our Payoffs Playing Best Responses According to $M$}
\end{table}

The sequence shown in the first three rounds will repeat indefinitely as long as we continue to play the same best responses with respect to $M$, leading us to lose or tie in every round. 

Note that it is possible to construct very similar, larger games in which we lose to the Myopic Best Responder an arbitrarily large fraction of the time (depending on the number of actions in the game), since we only need to tie against the first action in the opponent's action ordering to avoid discovering any inconsistency. 

\section{Strategies for Beating Behaviorally Biased Opponents}
    \subsection{Myopic Best Responder} 
    
    \begin{algorithm} 
    \caption{Beat the Myopic Best Responder \label{alg_mbr}} 
    \begin{enumerate} 
        \item \textit{Learn best responses:} In round $r$ for $1 \leq r \leq n + 1$:
        \begin{itemize}
            \item If $r \leq n$, play action $a_r$. Otherwise (in round $n+1$), play an arbitrary action.
            \item If r > 1, record the opponent’s action as a best response to our action in round $r-1$.
        \end{itemize}
        \item In round $r = n+2, ...$: 
        \begin{itemize}
            \item Predict the opponent's action as the recorded best response to our action in round $r-1$.
            \item Play the recorded best response to the predicted action.
        \end{itemize}
    \end{enumerate}
    \end{algorithm}
    
    \begin{theorem}
    Playing Algorithm \ref{alg_mbr} against the Myopic Best Responder in a permissible game (Definition \ref{def:permissible_game}) results in winning every round after the first $n+1$ rounds. 
    \end{theorem}
    \begin{proof}
    The Myopic Best Responder plays a best response to our previous action, so we record a correct best response to each action during the first $n+1$ rounds. The Myopic Best Responder always plays the same best response (the first one in its action ordering) following any given action, so we correctly predict the action it will play from round $n+2$ onward. Therefore we win every round from round $n+2$ onward, since we correctly predict the opponent's action and play a valid best response to it.
    \end{proof}

    \subsection{Gambler's Fallacy Opponent}

    \begin{algorithm}[H] 
    \caption{Beat the Gambler's Fallacy Opponent \label{alg_gambler}}
    \begin{enumerate}
        \item \textit{Learn a best response to $a_n$:} For the first $n$ rounds, play action $a_r$ during round $r$. Record the opponent's action in round $n$ as a best response to action $a_n$, call it $BR(a_n)$.
        \item \textit{Learn a best response to $BR(a_n)$:} During the next $n$ rounds ($n < r \leq 2n$), play each of the $n$ actions in any order as long as $BR(a_n)$ is last. Record the opponent's action in round $2n$ as a best response to $BR(a_n)$, call it $BR(BR(a_n))$. 
        \item \textit{Get the opponent to play $BR(a_n)$:} During the next $n-1$ rounds ($2n < r < 3n$), play each action except for $a_n$.
        \item For $r \geq 3n$, play $BR(BR(a_n))$.
    \end{enumerate}
    
    \end{algorithm}
    \begin{theorem}
        Playing Algorithm \ref{alg_gambler} against the Gambler's Fallacy opponent in a permissible game (Definition $\ref{def:permissible_game}$) results in winning every round from round $3n$ onward. 
    \end{theorem}
    \begin{proof} Recall that the Gambler's Fallacy opponent plays the best response to our ``most overdue'' (least-frequently-played) action. During round $n$, $a_n$ is the most overdue action since it is the only unplayed action, so the opponent's action in round $n$, $BR(a_n)$, is a valid best response to $a_n$. During round $2n$, $BR(a_n)$ is the most overdue as every other action has been played twice while it has been played once, so the opponent's action in that round, $BR(BR(a_n))$, is a valid best response to it. During round $3n$, $a_n$ is the most overdue since all other actions have been played three times while it has been played twice. Since the Gambler's Fallacy opponent always plays the same best response to a given action and it previously chose $BR(a_n)$ as its best response to $a_n$, it will play $BR(a_n)$. Since we play $BR(BR(a_n))$, which we showed above is a best response to $BR(a_n)$, we win round $3n$. Now, note that $a_n$ and $BR(BR(a_n))$ cannot be the same action: that would imply that $BR(a_n)$ both beats and is beaten by $a_n$. Therefore, $a_n$ remains the ``most overdue'' action, and the opponent will play $BR(a_n)$ again in the next round. We therefore win every round from round $3n$ onward by always playing $BR(BR(a_n))$, as the opponent continues to play $BR(a_n)$ as long as we do so. 
    \end{proof}
        
    \subsection{Win-Stay, Lose-Shift Opponent} \label{sec_win_stay_lose_shift}
    Recall that the Win-Stay Lose-Shift opponent plays the same action immediately following a win, and switches to the next action in its action ordering immediately following a loss. The Tie-Shift variant of this opponent treats a tie like a loss and shifts, and the Tie-Stay variant treats a tie like a win and stays.

     \subsubsection{Variant: Tie-Shift}

     \begin{algorithm}[H] 
         \caption{Beat the Tie-Shift Variant of Win-Stay, Lose-Shift Opponent \label{alg_tie_shift}}
         \begin{enumerate}
             \item \textit{Learn the opponent's action ordering:} Repeat the following n-1 times: 
             \begin{itemize} 
                \item Play each action in succession until the opponent switches actions. 
                \item Record the order the opponent plays actions as its action ordering.
            \end{itemize}
             \item \textit{Learn best responses:} For $1 \leq i \leq n$:
             \begin{itemize}
                 \item Play action $a_i$ until the opponent plays any action twice in a row, or until the opponent has played every action once.
                 \item Afterward, record the opponent's last action as a best response to $a_i$.
             \end{itemize}
             \textit{Note: in the final round of this step, we play $a_n$ and the opponent plays a best response to $a_n$, so the opponent will repeat their action in the next round.} 
             \item Play the recorded best response to the opponent's last action. 
             \item For all remaining rounds:
             \begin{itemize}
                 \item Predict the opponent's next action as the action following their last-played action in the recorded action order (looping around if at the end). 
                 \item Play the recorded best response to the predicted action. 
            \end{itemize}
         \end{enumerate}
     \end{algorithm}

     \begin{theorem}
        Playing Algorithm \ref{alg_tie_shift} against the Tie-Shift variant of the Win-Stay Lose-Shift opponent in a permissible game (Definition $\ref{def:permissible_game}$) wins all but at most $2n^2 - 2n + 1$ rounds.
     \end{theorem}
     \begin{proof}
        In the first phase, we record the correct action ordering: the opponent starts by playing the first action in their action ordering and always shifts to the next action in the ordering, so by observing $n-1$ shifts we observe all $n$ actions in the correct order. Because we play each action in succession against the opponent's current action, they are guaranteed to shift after at most $n-1$ rounds, since their action must tie with itself and lose to at least one other action.
     
        Now, we claim that we learn a best response to each action in phase 2. Proof: For each action $a_i$, we observe an action the opponent plays twice in response to it or we observe that the opponent plays each action once in response to it. If the opponent plays the same action twice, the opponent's action beats $a_i$ because the opponent only repeats the same action immediately following a win. If the opponent has played every action once against our action, $a_i$ must not have been beaten by any of the first $n-1$ actions played by the opponent, since the opponent shifted after playing each of those against $a_i$. Since each action must be beaten by at least one other action, the opponent's last action must beat $a_i$. Therefore, we correctly record an action which beats each action $a_i$. In our setting, where there is only one payoff value corresponding to a win, any action which beats $a_i$ is a best response to $a_i$.

        In the final round of phase 2, we played $a_n$ and the opponent played a best response to $a_n$. Therefore the opponent will repeat their action again in the next round, since they stay following a win, and we win that round by playing the recorded best response to their repeated action in phase 3.

        At the start of phase 4, we correctly predict the opponent's next action: we know we won the last round, so we know that the opponent will shift to the next action in its action order (which we have recorded correctly, as shown above). We showed above that we correctly recorded a best response to each action, so we win by playing the recorded best response to the predicted action. At the start of the next round, the same conditions hold (and will hold after each round), so we will win all subsequent rounds.
     
        Phase 1 (the process of learning the action ordering) requires $n-1$ rounds, each of which could incur up to $n-1$ losses or ties (each action loses to at least one action), for $n^2 - 2n + 1$ losses or ties overall. Phase 2 (the process of learning best responses) requires $n$ rounds, and each can incur up to $n$ losses or ties (note that there can only be up to 2 losses per round however, since the opponent would repeat its action in that case and the round would end), for $n^2$ losses or ties overall. We always win during the remaining phases, so the total number of losses or ties is bounded by $2n^2 - 2n + 1$.
     \end{proof}

    \subsubsection{Variant: Tie-Stay} 
    The main difference in beating the Tie-Stay variant compared to the Tie-Shift variant is that we can find a best response to each action by simply playing each action in succession until the opponent switches, since they only switch following a loss. For the algorithm, theorem and proof see $\ref{sec_tie_stay}$ in the Appendix.

    \subsection{Follow-the-Leader Opponent} \label{sec_ftl}

    Recall that the Follow-the-Leader opponent plays the best action in retrospect, defined as the action that would have achieved the highest payoff against our entire history of play. For this opponent, our strategy will be to learn a best response to each action, and then use the well-known ellipsoid algorithm to predict the opponent's actions while playing best responses to the predicted actions. 

    \paragraph*{Using Ellipsoid for Prediction} We use ellipsoid to find a point $M'$ in $n^2$-dimensional space roughly corresponding to the game payoff matrix $M$ (each of the coordinates of $M'$ approximates a payoff value). We have variables $m'_{ij}$ for each coordinate of $M'$ representing the payoff of $a_i$ against $a_j$, and coefficients $c_j$ denoting the number of times we have played $a_j$ (with $i,j \in [n]$). We can then use $M'$ to predict the opponent's best action in retrospect: $a_i$ where $i$ is $\arg\max_{i} \sum_{j} c_j m'_{ij}$. 

    The ellipsoid algorithm requires implementing a separation oracle that returns a violated constraint each time the algorithm makes a mistake. In this case, the oracle is inherent to the opponent's play: whenever a prediction error is made, the net score of the action the opponent played against our entire history must be greater than or equal to the net score of the action we predicted. However, to avoid separately handling ties and the opponent's unknown action ordering, we need these inequalities to be strict. Without loss of generality, we assume there are tiny values added to each true payoff value such that there is a slightly higher payoff for tied actions earlier in the opponent's action ordering. These values should be small relative to any actual payoff differences and inversely proportional to the number of rounds so that they cannot accumulate over time to outweigh any actual payoff differences. Given this assumption, whenever the opponent actually played $a_{i^*}$ and we predicted $a_{i'}$, we observe that $\sum_{j \in [n]} c_j m'_{i^* j} > \sum_{j \in [n]} c_j m'_{i'j}$.

    \begin{algorithm}[H] 
         \caption{Beat the Follow-the-Leader Opponent \label{alg_ftl}}
         \begin{enumerate}
             \item \textit{Learn Best Responses:} For $1 \leq i \leq n$: 
             \begin{itemize}
                 \item If $i \neq n$, play action $a_i$ $3^{i-1}$ times. Otherwise, play $a_i$ $3^{i-1} + 1$ times.
                 \item Record the opponent's action in round immediately after playing $a_i$ for the $3^{i-1}$th time as a best response to $a_i$. 
             \end{itemize}
             \item In all remaining rounds:
             \begin{itemize}
                 \item Predict the opponent’s next action using the estimate of the game matrix produced by
the ellipsoid algorithm, breaking ties arbitrarily.
                 \item Play the recorded best responses to the predicted action.
             \end{itemize}
         \end{enumerate}
     \end{algorithm}

    \begin{theorem}
        Playing Algorithm \ref{alg_ftl} against the Follow-the-Leader opponent in a permissible game (Definition $\ref{def:permissible_game}$) wins all but $O(3^n + n^4 \log(nr))$ rounds over $r$ rounds of play.
    \end{theorem}
    
    \begin{proof} 
    First, we show we correctly record a best response to each action in phase 1 of the algorithm. The high level idea is that we learn a best response to each action $a_i$ by playing each action enough times in a row that the opponent must play an action which beats it (which, in our setting with only a single payoff value corresponding to a win, is a best response). Note that such an action must exist since in a permissible game, each action is beaten by at least one other action. Let $BR(a_i)$ denote the best response to $a_i$ that the opponent will eventually play. For any action $b$ which does $\textit{not}$ beat $a_i$, the highest possible payoff over our entire history would be achieved if $b$ beat every action in our history besides $a_i$, and tied against $a_i$. Meanwhile, the lowest possible score for $BR(a_i)$ would be if $BR(a_i)$ lost to every action in our history besides $a_i$ (while of course winning against $a_i$). It is also possible that $b$ precedes $BR(a_i)$ in the opponent's action ordering, so we need to play $a_i$ enough times that the lowest possible score of $BR(a_i)$ is strictly greater than the highest possible score of any such action $b$. To achieve this, we need to play $a_i$ twice as many times as the total number of previous rounds (to equalize the worst-case scores of $-1$ that $BR(a_i)$ got in all previous rounds with $b$'s $+1$s) plus 1 (to break the tie in case $b$ precedes $BR(a_i)$ in the opponent's action ordering). When doing this for each action in succession, we should therefore play the first action 1 time, second 3 times, followed by 9, 27, etc., which is $3^{i-1}$ times for the $i$th action as indicated in the algorithm. Note that we play $a_n$ one extra time to allow observation of $BR(a_n)$ in the following round, before moving on to phase 2. 

    In first phase, there are $\sum_i 3^{i-1} + 1 = \frac{3^n -1}{2} +1$ rounds, so it incurs at most that many losses or ties. Then, since we correctly recorded best responses, we only tie or lose during phase 2 if the ellipsoid algorithm makes a prediction mistake. The ellipsoid algorithm makes at most $O(n^4 \log(nr))$ mistakes, where r is the total number of rounds played (see $\ref{sec_ellipsoid_mistake_bounds}$ in the Appendix for proof), so the total number of losses or ties is bounded by $O(3^n + n^4 \log(nr))$. \end{proof}

    Note that while the bound on the number of losses and ties we incur is exponential, the runtime can be considered efficient since choosing which action to play in each round is efficient; this is an improvement over the simple general prediction algorithm we considered in section \ref{sec_general_prediction}, which requires considering an exponential number of game matrices to choose which action to play in a single round. We also consider a limited-history variant of the Follow-the-Leader opponent below, against which we can achieve a polynomial bound on losses and ties.

    \subsubsection{Variant: Limited History}
    Recall that the Limited History variant of the Follow-the-Leader opponent plays the action that would have achieved the highest net payoff against the last $r$ rounds of our play. We can use an almost identical algorithm to beat this variant, except that we can force the opponent to play a best response to any action $a_i$ by playing $a_i$ just $r$ times. This allows us to beat this variant with a polynomial bound on ties and losses. For the full algorithm, theorem, and proof, see the \ref{sec_ftl_limited_history} in the Appendix. 
    
    \subsection{Highest Average Payoff Opponent}
    The Highest Average Payoff opponent plays the action that has achieved the highest average payoff over the times they have played it. We discuss this opponent in $\ref{sec_highest_avg_payoff}$.

\section{Generalizing}
\subsection{Other Behaviorally-Biased Strategies} \label{sec_generalizing_other_opponents}
A natural question to ask is, what kinds of behaviorally-biased strategies can we exploit to win nearly every round of a permissible game (Definition \ref{def:permissible_game})? Clearly, if we can predict the opponent's actions and learn best responses to those actions the opponent plays, we can achieve this goal.

\paragraph*{Predicting the Opponent’s Actions} \label{sec_general_pred_multiple_strategies}
Given the opponent's known biased strategy and a finite set of $t$ possible tie-breaking mechanisms, which together define the family $\mathbb{S}$ of possible deterministic strategies the opponent could be using, we can again predict the opponent's actions using the simple halving algorithm described in Section \ref{sec_general_prediction} (Algorithm \ref{alg_prediction_by_halving}). As $|\mathbb{S}| = t$, from Theorem $\ref{thm_halving_bounds}$, this will incur at most $O(n^2 + log(t))$ prediction mistakes and take time $(O(3^{n^2}tR))$ per round  (where $R$ is the time in which the opponent's strategy chooses an action given the game matrix and history). 

\paragraph*{Learning Best Responses}
The more challenging question is which deterministic strategies we can exploit to learn best responses for any action the opponent might play. Note that we cannot do this for every deterministic strategy; consider a very simple opponent who always plays the same action $a$. In this case, the opponent's play does not reveal any information about a best response to $a$. 

\begin{observation}
    It is not always possible to exploit a known deterministic strategy to guarantee winning any rounds.
\end{observation}

One sufficient condition which many natural behavioral strategies may have is that after we play an action ``enough'' times, the opponent will eventually play a best response to it. ``Enough'' needs to be some number we can determine; for example, a constant $c$, or some multiple $b$ of the number of previous rounds $r$, beginning with $b$ when there are no previous rounds. Then we can simply play each action in succession ``enough'' times, observing the opponent's action in the following round, to learn a best response to each action. In this case, the bound on the number of losses or ties during the process of learning best responses would be $cn + 1$ or $b(b+1)^{n-1} + 1$, respectively (assuming we learn best responses at the start of the game). 

\begin{observation}
    If the opponent plays a best response to an action if it is played $c$ times in a row, where $c$ is some constant, we incur at most $cn + 1$ losses or ties during the process of learning best responses. If the opponent plays a best response to an action if it is played a multiple $b$ of the number of previous rounds $r$, we incur up to $b(b+1)^{n-1} + 1$ losses or ties while learning best responses. 
\end{observation}

Note that not every reasonable behavioral strategy has this property; several of the biased opponents we discussed earlier do not: the Gambler's Fallacy opponent, the Tie-Stay variant of Win-Stay Lose-Shift opponent, and the Highest Average Payoff opponent. 

\subsection{Exploiting an Unknown Strategy from a Known Set of Strategies}

Another natural extension is to consider the scenario where we know the opponent uses some strategy from a known set of biased strategies $\mathbb{B}$, but we don't know which one. Again, we assume strategies in $\mathbb{B}$ are deterministic when parameterized with one of $t$ possible deterministic tie-breaking mechanisms.

\paragraph*{Predicting the Opponent’s Actions}
We can once again predict generally using the same halving algorithm from earlier (Algorithm $\ref{alg_prediction_by_halving}$). We only need to make a trivial modification to the input family $\mathbb{S}$ of deterministic strategies: instead of defining $\mathbb{S}$ to be the set of all possible pairs of a single known biased strategy with its possible tie-breaking mechanisms, we can define it to be the set of all possible pairs among $\textit{all}$ possible biased strategies in $\mathbb{B}$ and all $t$ possible tie-breaking mechanisms. Then since $|\mathbb{S}| = |\mathbb{B}|t$, from Theorem $\ref{thm_halving_bounds}$, Algorithm $\ref{alg_prediction_by_halving}$ makes up to $O(n^2 + \log(|\mathbb{B}|) + \log(t))$ prediction mistakes in $O(3^{n^2}|\mathbb{B}|tR)$ time per round (where $O(R)$ is the time in which any strategy in $\mathbb{B}$ chooses an action given the game matrix and history).

\paragraph*{Learning Best Responses}
The challenge, again, is to learn best responses. Generally, any time we can use the same algorithm to learn best responses given any strategy in $\mathbb{B}$, this extension is trivial. For example if all strategies in $\mathbb{B}$ have the natural property discussed above (in $\ref{sec_generalizing_other_opponents}$) that after we play an action ``enough'' times, the opponent will eventually play a best response to it, we can use the same algorithm from $\ref{sec_generalizing_other_opponents}$ to learn best responses (using the maximum value required by any strategy as ``enough''). 

\begin{observation}
    If all strategies in $\mathbb{B}$ have property that the opponent plays a best response to an action if it is played ``enough'' times in a row, and we know the maximum upper bound on ``enough'', we can learn best responses as described in $\ref{sec_generalizing_other_opponents}$. 
\end{observation}

If it is not possible to use the same algorithm to learn best responses given any strategy in $\mathbb{B}$, it might instead be possible to first distinguish between strategies which need different best-response algorithms, and then learn best responses. As a very simple example, if we know the opponent uses one of the Tie-Stay variant or the Tie-Shift variant of the Win-Stay, Lose-Shift strategy, we cannot use the strategy mentioned above because the Tie-Stay variant does not have the property in question. However, we can distinguish between the two strategies by checking whether the opponent ``stays'' after we play the same action as it (since a permissible game is symmetric, this must be a tie. We can check this within $n$ rounds by playing the same action $a$ until either the opponent switches to playing $a$, or stays on some other action $b$. If the opponent stays on $b$, we can then play $b$ to check their tie behavior). Then, knowing our opponent's strategy, we can use the corresponding algorithm from Section \ref{sec_win_stay_lose_shift} to learn best responses (and in this case, to predict efficiently as well). 

Note, however, that it is not always possible to distinguish between deterministic strategies. As a simple example, consider playing a 3-action game against either the Myopic Best Responder, who plays a best response to our last action, or another opponent, call them Myopic Worst Responder, who plays a worst response to our last action. No matter how we play, the Myopic Best Responder with respect to rock-paper-scissors behaves identically to the Myopic Worst Responder with respect to a reverse version of rock-paper-scissors where the wins and losses are switched. Since both games are possible, we cannot distinguish whether the opponent is playing Myopic Best Response or Myopic Worst Response. Moreover, best responses with respect to one scenario are worst responses with respect to the other. Either of these opponents, if known, would be easy to exploit. However, in this case we cannot guarantee an average payoff greater than zero: if we play a best responses with respect to either opponent in a given round, it is a loss with respect to the other so our expected payoff for that round is 0. 

\begin{observation}
    There exist collections of pairs of game matrices and strategies such that it is not possible to guarantee a positive average payoff, even if we could exploit any of the individual strategies if known to win nearly every round. 
\end{observation}

\section{Future Work}
In future work, it would be interesting to further characterize which behaviorally-biased strategies can be exploited to win nearly every round without observing payoffs. In that vein, one clear next step would be to explore the exploitability of behaviorally-motivated probabilistic strategies. In the probabilistic setting it may not be possible to win even most of the time; for example, probabilistic strategies could be regret-minimizing, in which case the best we can hope to achieve is the minimax value of the game. However, we could aim to get bounds on our performance relative to the best possible strategy given full information about the game matrix, payoffs, and the opponent’s strategy. Another direction would be considering more complex games, for example exploring extensive-form games. 



\bibliography{bib}

\begin{thebibliography}{10}

\bibitem{AGARWAL2023119117}
Deepesh Agarwal and Balasubramaniam Natarajan.
\newblock Tracking and handling behavioral biases in active learning frameworks.
\newblock {\em Information Sciences}, 641:119117, 2023.
\newblock URL: \url{https://www.sciencedirect.com/science/article/pii/S0020025523007028}, \href {https://doi.org/10.1016/j.ins.2023.119117} {\path{doi:10.1016/j.ins.2023.119117}}.

\bibitem{auer2002nonstochastic}
Peter Auer, Nicol{\`{o}} Cesa{-}Bianchi, Yoav Freund, and Robert~E. Schapire.
\newblock The nonstochastic multiarmed bandit problem.
\newblock {\em {SIAM} Journal on Computing}, 32(1):48--77, 2002.
\newblock \href {https://doi.org/10.1137/S0097539701398375} {\path{doi:10.1137/S0097539701398375}}.

\bibitem{2011behavioralgttext}
C.F. Camerer.
\newblock {\em Behavioral Game Theory: Experiments in Strategic Interaction}.
\newblock The Roundtable Series in Behavioral Economics. Princeton University Press, 2011.
\newblock URL: \url{https://books.google.com/books?id=cr_Xg7cRvdcC}.

\bibitem{cartwright2018behavioral}
Edward Cartwright.
\newblock {\em Behavioral economics}.
\newblock Routledge, 2018.

\bibitem{predlearninggamesbook}
Nicol{\`{o}} Cesa{-}Bianchi and G{\'{a}}bor Lugosi.
\newblock {\em Prediction, learning, and games}.
\newblock Cambridge University Press, 2006.
\newblock \href {https://doi.org/10.1017/CBO9780511546921} {\path{doi:10.1017/CBO9780511546921}}.

\bibitem{securitygames}
Nicolas Christin.
\newblock Network security games: Combining game theory, behavioral economics, and network measurements.
\newblock In John~S. Baras, Jonathan Katz, and Eitan Altman, editors, {\em Decision and Game Theory for Security}, pages 4--6, Berlin, Heidelberg, 2011. Springer Berlin Heidelberg.
\newblock \href {https://doi.org/10.1007/978-3-642-25280-8_2} {\path{doi:10.1007/978-3-642-25280-8_2}}.

\bibitem{DEVETAG2008364}
Giovanna Devetag and Massimo Warglien.
\newblock Playing the wrong game: An experimental analysis of relational complexity and strategic misrepresentation.
\newblock {\em Games and Economic Behavior}, 62(2):364--382, 2008.
\newblock URL: \url{https://www.sciencedirect.com/science/article/pii/S0899825607001133}, \href {https://doi.org/10.1016/j.geb.2007.05.007} {\path{doi:10.1016/j.geb.2007.05.007}}.

\bibitem{feldman2010playing}
Michal Feldman, Adam Kalai, and Moshe Tennenholtz.
\newblock Playing games without observing payoffs.
\newblock In {\em Innovations in Computer Science - {ICS} 2010, Tsinghua University, Beijing, China, January 5-7, 2010. Proceedings}, pages 106--110. Tsinghua University Press, 2010.
\newblock URL: \url{http://conference.iiis.tsinghua.edu.cn/ICS2010/content/papers/9.html}.

\bibitem{halpernalgrationality}
Joseph~Y. Halpern and Rafael Pass.
\newblock Algorithmic rationality: Game theory with costly computation.
\newblock {\em Journal of Economic Theory}, 156:246--268, 2015.
\newblock URL: \url{https://doi.org/10.1016/j.jet.2014.04.007}, \href {https://doi.org/10.1016/J.JET.2014.04.007} {\path{doi:10.1016/J.JET.2014.04.007}}.

\bibitem{kahneman2011thinking}
Daniel Kahneman.
\newblock {\em Thinking, Fast and Slow}.
\newblock Farrar, Straus and Giroux, New York, 2011.

\bibitem{littlestone1988learning}
Nick Littlestone.
\newblock Learning quickly when irrelevant attributes abound: {A} new linear-threshold algorithm.
\newblock {\em Machine Learning}, 2(4):285--318, 1988.
\newblock \href {https://doi.org/10.1007/BF00116827} {\path{doi:10.1007/BF00116827}}.

\bibitem{learningexploitationbias}
John~M. McNamara, Alasdair~I. Houston, and Olof Leimar.
\newblock Learning, exploitation and bias in games.
\newblock {\em PLOS ONE}, 16(2):1--14, 02 2021.
\newblock \href {https://doi.org/10.1371/journal.pone.0246588} {\path{doi:10.1371/journal.pone.0246588}}.

\bibitem{doiStrategyWinstay}
Martin Nowak and Karl Sigmund.
\newblock A strategy of win-stay, lose-shift that outperforms tit-for-tat in the prisoner’s dilemma game.
\newblock {\em Nature}, 364(6432):56–58, July 1993.
\newblock \href {https://doi.org/10.1038/364056a0} {\path{doi:10.1038/364056a0}}.

\bibitem{Osborne1994}
Martin~J. Osborne and Ariel Rubinstein.
\newblock {\em A course in game theory}.
\newblock The MIT Press, Cambridge, USA, 1994.
\newblock electronic edition.
\newblock URL: \url{https://arielrubinstein.tau.ac.il/books/GT.pdf}.

\bibitem{bams/1183517370}
Herbert Robbins.
\newblock {Some aspects of the sequential design of experiments}.
\newblock {\em Bulletin of the American Mathematical Society}, 58(5):527 -- 535, 1952.
\newblock \href {https://doi.org/10.1090/s0002-9904-1952-09620-8} {\path{doi:10.1090/s0002-9904-1952-09620-8}}.

\bibitem{shalev2012online}
Shai Shalev{-}Shwartz.
\newblock Online learning and online convex optimization.
\newblock {\em Foundations and Trends in Machine Learning}, 4(2):107--194, 2012.
\newblock \href {https://doi.org/10.1561/2200000018} {\path{doi:10.1561/2200000018}}.

\bibitem{sutton2018reinforcement}
Richard~S Sutton and Andrew~G Barto.
\newblock {\em Reinforcement learning: An introduction}.
\newblock MIT press, 2018.

\bibitem{vempalaellipsoidlecturenotes}
Santosh Vempala.
\newblock Algorithmic convex geometry.
\newblock \url{https://faculty.cc.gatech.edu/~vempala/acg/notes.pdf}, 2008.
\newblock [Accessed 06-Sep-2023].

\bibitem{Wang_2014}
Zhijian Wang, Bin Xu, and Hai-Jun Zhou.
\newblock Social cycling and conditional responses in the rock-paper-scissors game.
\newblock {\em Scientific Reports}, 4(1), jul 2014.
\newblock \href {https://doi.org/10.1038/srep05830} {\path{doi:10.1038/srep05830}}.

\end{thebibliography}

\appendix
\section{Appendix}

 \subsection{Win-Stay Lose-Shift Variant: Tie-Stay} \label{sec_tie_stay}

    \begin{algorithm}[H] 
    \caption{Beat the Tie-Stay Variant of Win-Stay, Lose-Shift Opponent \label{alg_tie_stay}}
    \begin{enumerate}
        \item \textit{Learn best responses and the opponent's action order:} Repeat the following $n$ times:
            \begin{itemize}
                \item Play each action in succession until the opponent switches from their current action to a new action. 
                \item Record our last action before the switch as a best response to the opponent's action before the switch. 
                \item Record the order the opponent plays actions as its action ordering (stopping once we've recorded all $n$ actions).
            \end{itemize}
            \textit{Note: at the end of this step, the opponent will have just played the first action in its action order, and we will have just played some action $a$.} 
        \item \textit{Get into a state from which we can begin to predict the opponent's actions:} Repeatedly play action $a$ until the opponent plays the same action, call it $b$, twice in a row. \textit{Note: the next action the opponent plays will also be b.}
        
        \item Play the recorded best response to $b$.
        \item In all remaining rounds:
        \begin{itemize}
            \item Predict the opponent's next action by finding the action immediately following the last action they played in the action ordering we recorded (looping around if at the end).
            \item Play the recorded best response to the predicted action. 
        \end{itemize} 
    \end{enumerate}
    \end{algorithm}

       \begin{theorem}
    Playing Algorithm \ref{alg_tie_stay} against the Tie-Stay variant of the Win-Stay Lose-Shift opponent in a permissible game (Definition $\ref{def:permissible_game}$) wins all but at most $n^2 - n + 2$ rounds.
    
    \end{theorem}
     \begin{proof}
    Recall that the Tie-Stay variant of the Win-Stay Lose-Shift opponent plays the same action immediately following a win or a tie, and switches to the next action in its action ordering immediately following a loss. Because each action is beaten by at least one other action, the opponent must switch after one of the actions we play in response to its current action in phase 1 (since we play each action in succession). If the opponent shifts to play a new action in round $r$, we won round $r-1$, so the best response we record is correct. Since the opponent always shifts to the next action in its action ordering, over $n$ shifts, they will shift through every action before finally shifting back to the first action in its action ordering. Through this process we therefore record a best response to every action, and we record the correct action ordering in step 1. We play at most $n$ actions in response to each action the opponent plays to make it shift, and one of them must beat the opponent's action, so we incur no more than $n(n-1)$ losses or ties in this phase. 

     At the end of the above, the opponent will have shifted back to the first action of their action ordering. However, we don't generally know in advance when this shift will occur, so the action we played (some $a$) may or may not have been a best response to it and we therefore can't be sure whether the opponent will stay or shift. By playing $a$ repeatedly, the opponent will eventually play some action twice; this will happen within $n$ rounds since each action is beaten by at least one other action and ties with itself. However, while the opponent is shifting we must be winning, so we only incur 2 losses or ties here, when the opponent finally plays some action $b$ twice in a row.

     Since we finished round 2 with the opponent tying or winning, they will repeat their action ($b$) again in the next round; we win that round by playing the recorded best response to $b$ (which we showed was correct) in phase 3.
     
     At the start of phase 4, we predict correctly because we won the previous round and know the opponent will shift to the next action in its order (which we showed we recorded correctly). We showed earlier that we recorded correct best responses, so we win by playing the recorded best response to the predicted action. The same conditions hold for every remaining round, so we win all remaining rounds. 
     
     The total the number of losses or ties is bounded by the sum of those in the first two phases, which is at most $n^2 - n + 2$.
     \end{proof}

\subsection{Follow-the-Leader Variant: Limited History} \label{sec_ftl_limited_history}
    Recall that the limited-history variant of the Follow-the-Leader opponent plays the action that would have achieved the highest net payoff against the last $r$ rounds of our play. Note that we assume $r > 0$. If $r = 0$, the opponent would simply play the same action every round (the first action in their action ordering) regardless of our play. In that case, their play would not reveal any information about best responses, so the best we could do is tie in most rounds by playing the same action as the opponent.

    We again use the ellipsoid algorithm for prediction, essentially the same way as described for the unlimited-history variant of this opponent in $\ref{sec_ftl}$. In this case, the inequality we observe when a mistake is made is only relative to the past r rounds.

        \begin{algorithm}[H] 
         \caption{Beat the Limited-History Variant of the Follow-the-Leader Opponent \label{alg_ftl_limited}}
         \begin{enumerate}
             \item \textit{Learn Best Responses:} For $1 \leq i \leq n$:
             \begin{itemize}
                 \item If $i \neq n$, play action $a_i$ $r$ times. Otherwise, play $a_i$ r+1 times. 
                 \item Record the opponent's action during round $ir + 1$ as a best response to action $a_i$. 
             \end{itemize}
             \item In all remaining rounds:
             \begin{itemize}
                 \item Predict the opponent’s next action using the estimate of the game matrix produced by
the ellipsoid algorithm, breaking ties arbitrarily.
                 \item Play the recorded best responses to the predicted action.
             \end{itemize}
         \end{enumerate}
     \end{algorithm}

    \begin{theorem}
        Playing Algorithm \ref{alg_ftl_limited} against the limited-history variant of the Follow-the-Leader opponent in a permissible game (Definition $\ref{def:permissible_game}$) wins all but $O(n^4 \log(nr) + nr)$ rounds, where $r$ is the number of rounds the opponent includes in its limited history.
    \end{theorem}
    
    \begin{proof} 
    First, we prove we correctly record a best response to each action in phase 1 of the algorithm. After we play any action $a_i$ $r$ times in a row, the opponent must play a best response to $a_i$ in the following round: the opponent plays the action which would have been best over the previous $r$ rounds in which we played $a_i$ every time, which must be a best response to action $a_i$. Since we play each action $r$ times in a row (except $a_n$ which we play $r+1$ times) and iterate through actions in ascending order, the opponent must play a best response to $a_i$ in round $ir+1$. Note that we play $a_n$ one extra time to allow for observing the best response in round $nr + 1$ before moving on to phase 2.  
    
    Since there are $nr + 1$ rounds in phase 1, we can lose or tie at most that many times in phase 1. Then, since we correctly recorded best responses, we only tie or lose during phase 2 if the ellipsoid algorithm makes a prediction mistake. The ellipsoid algorithm makes at most $O(n^4 \log(nr))$ mistakes, where $r$ is again the history parameter rather than the total number of rounds played (see $\ref{sec_ellipsoid_mistake_bounds}$ for proof). Therefore, the total number of mistakes is bounded by $O(n^4 \log(nr) + nr)$. \end{proof} 

    \subsection{Ellipsoid Mistake Bounds} \label{sec_ellipsoid_mistake_bounds}
    In general, the ellipsoid algorithm requires up to $O(d^2\log \frac{R_{outer}}{R_{inner}})$ iterations, where $d$ is the dimension of the space, $R_{outer}$ is the radius of the initial ellipsoid which encompasses the entire solution space, and $R_{inner}$ is radius of the largest ellipsoid contained in the feasible space \cite{vempalaellipsoidlecturenotes}. In this case, the dimension is $n^2$. Since payoffs are all between -1 and 1, the solution space is contained in an $n^2$-dimensional ball of radius $n$. Finally, we will determine $R_{inner}$, which differs between the Follow-the-Leader opponent and the Highest-Average-Payoff opponent (although we will get overall mistake bounds of the same order in either case). As noted in \ref{sec_ftl}, in order to get strict inequalities, we assume without loss of generality that there are tiny values added to each true payoff value such that there is a slightly higher payoff for tied actions earlier in the opponent’s action ordering. Note that we do not actually add these values, which would require knowledge of the true payoffs and the opponent's action ordering; it is sufficient that such values exist. The appropriate values differ based on the opponent and help us determine $R_{inner}$ in each case. 
    
    \paragraph*{Follow-the-Leader Opponent}
    For the Follow-the-Leader opponent, the values added to each true payoff should be small relative to any actual payoff differences and inversely proportional to the number of rounds so that they cannot accumulate over time to outweigh any actual payoff differences. In the limited history case, they only need to be inversely proportional to the number of rounds which the opponent considers. More concretely, consider adding $\frac{1}{ir}$ to the opponent's payoffs for playing the $i$th action in their action ordering (i.e., to all payoffs in the $i$th row of the game matrix if the opponent is the row player), where $r$ is the total number of rounds for the unlimited-history variant or the length of history considered in the limited-history variant. Then, it's easy to show that the net score of the best action is at least $\frac{1}{n^2 - n}$ higher than any other action. The solution point could therefore shift up to $\frac{1}{r(n^2 - n)}$ (not inclusive) in any direction without altering the induced predictions, so the feasible space contains an $n^2$-dimensional ball of radius $\frac{1}{r(n^2 - n)} - \epsilon$, for some arbitrarily small $\epsilon$. Overall, then, the number of iterations required by ellipsoid in this case is bounded by $O(n^4\log (nr)$).

    \paragraph*{Highest Average Payoff Opponent}
     For the Follow-the-Leader opponent, the values look somewhat different because the scores are averaged as part of the opponent's strategy. Consider adding  $\frac{1}{ir^2}$ to the opponent's payoffs for playing the $i$th action in their action ordering. It's easy to verify that the average score of the best action is at least $\frac{1}{r^2(n^2 - n)}$ higher than that of any other action. The solution point could shift up to $\frac{1}{r^2(n^2 - n)}$ (not inclusive) in any direction without changing the induced predictions, so the feasible space contains an $n^2$-dimensional ball of radius $\frac{1}{r^2(n^2 - n)} - \epsilon$, for some arbitrarily small $\epsilon$. Overall, then, the number of iterations required by ellipsoid in this case is also bounded by $O(n^4\log (nr)$).

    \subsection{Highest Average Payoff Opponent} \label{sec_highest_avg_payoff}
     We assume the opponent initializes each action with an average payoff of 0. 

    For this opponent, our high-level strategy will be to learn a best response to each action, and then use the well-known ellipsoid algorithm to predict the opponent’s actions while playing best responses to the predicted actions.

    The ellipsoid algorithm will be used in essentially the same way as for the Follow-the-Leader opponent, in order to find a point $M'$ in $n^2$-dimensional space roughly corresponding to the game payoff matrix M. The main difference is that the inequality we get whenever there is a mistake will reflect that the opponent makes comparisons between average scores rather than net scores: whenever a mistake is made (i.e., the opponent actually played some action $a_{i^*}$ and we predicted another action $a_{i'}$), we observe that $\frac{\sum_{j \in [n]} c_{i^*j} m'_{i^*j}}{\sum_{j \in [n]} c_{i^*j}} > \frac{\sum_{j \in [n]} c_{i'j} m'_{i'j}}{\sum_{j \in [n]} c_{i'j}}$, where $c_{ij}$ is the number of rounds during which the opponent played $a_i$ and we played $a_j$, and $m'_{ij}$ is the payoff of $a_i$ against $a_j$ according to $M'$. See $\ref{sec_ellipsoid_mistake_bounds}$ for further details.  
    
    \begin{algorithm}[H] 
         \caption{Beat the Highest Average Payoff Opponent \label{alg_highest_avg_payoff}}
         \begin{enumerate}
         \item \textit{Learn best responses to the first $n-1$ actions in the opponent's order:} Repeat the following $n-1$ times: 
         \begin{itemize}
             \item Play each action in succession in ascending order (not starting over when the opponent switches to a new action, wrapping around from $a_n$ to $a_1$ if necessary), each up to as many times as the number of rounds the opponent has played its current action previously plus 1, until the opponent switches or we've played every action once. 
             \item Record our last action prior to the switch as a best response to the opponent's last action before the switch.
         \end{itemize}
         \textit{Note: At the end of this round, the opponent will have just played the final action in their action order, and we will have just played some action a.}
         \item \textit{Learn a best response to the $n$th action in the opponent's order:} Play each action in succession in ascending order, starting with the last action we played $a$ and wrapping around if necessary, each up to 3 times the number of rounds the opponent has played its current action so far (starting with once when the opponent has not played the $n$th action in their order before), until the opponent switches or we've played every action once.  Record our last action as a best response to the $n$th action in the opponent's order.
        \item In all remaining rounds:
        \begin{itemize}
            \item Predict the opponent's next action using the estimate of the game matrix produced by the ellipsoid algorithm, breaking ties arbitrarily.
            \item Play the recorded best response to the predicted action.
        \end{itemize} 
    \end{enumerate}
    \end{algorithm}

     \begin{theorem}
        Playing Algorithm \ref{alg_highest_avg_payoff} against the Highest Average Payoff opponent in a permissible game (Definition $\ref{def:permissible_game}$) wins all but $O(4^{n-1} + n^4\log(nr))$ rounds.
     \end{theorem}

    \begin{proof} First, we show that we correctly record a best response to each of the first $n-1$ actions in the opponent's action order in phase 1 of the algorithm. Recall that the opponent breaks ties among actions which have the same average payoff using a fixed order over actions, and denote the $i$th action in this order as $a_{f(i)}$ where $f$ maps an action's index in the opponent's ordering to its index in the game. Due to this tie-breaking strategy, we claim that the opponent will begin by playing $a_{f(1)}$ a number of times, followed by $a_{f(2)}$ some number of times and so on through $a_{f(n)}$ some number of times, switching after each of $a_{f(1)}, ..., a_{f(n-1)}$ exactly when it reaches a negative average payoff. Proof: at the start of the game, since all actions are initialized with the same average payoff of 0, the opponent will play $a_{f(1)}$. As long as the average payoff of $a_{f(1)}$ is non-negative, it at least as good as that of any other action, so the opponent will continue to play $a_f(1)$. Once it becomes negative, the opponent will switch to playing $a_f(2)$, since all unplayed actions have an average payoff of 0 and $a_{f(2)}$ is the earliest in the opponent's action ordering. More generally for 
    $a_f(i)$ when $i > 1$, as long as the average payoff of $a_{f(i)}$ remains non-negative, it is strictly better than all actions which precede it in the opponent's action order (every such action must have reached a negative average payoff for the opponent to have played $a_{f(i)}$ the first time, since $a_{f(i)}$ had an initial average payoff of 0 and is later in the action order). It is also at least as good as the average payoff of any actions which follow it in the action order (which all have an average payoff of 0 since they have not yet been played), so the opponent will continue to play  $a_{f(i)}$. As soon as the average payoff of $a_{f(i)}$ becomes negative, if $i < n$, the opponent will switch to the next action in the opponent's action order, which has not yet been played and is therefore the earliest action in the opponent's action order that is tied for the highest average payoff of $0$. 

    The round preceding such a switch must have been a loss for the opponent: its net payoff must have gone down during that round for the average payoff to go from non-negative to negative. We showed that the opponent will proceed through each of the $n$ actions in order when it switches during phase 1, so as long as we can reliably force the opponent to switch actions, we correctly record a best response to each of the first $n-1$ actions in phase 1.  

    To force the opponent to switch in phase 1, we play each action sequentially, up to as many times as the number of rounds the opponent has played its current action previously plus 1, until the opponent switches: we play the first action 1 time, the second 2 times, followed by 4, 8, etc. ($2^{i-1}$ times for the $i$th action). The net score of the opponent's current action is at most the number of times it has played it so far (if it won every previous round), so if our current action beats it, this would be enough to bring its net score to -1 in the worst case. Since at least one action beats the opponent's action, the opponent must switch after one of the actions we play in this succession. This process of getting the opponent to switch actions takes at most $2^n-1$ rounds total, so there are up to that many losses or ties per action, for $(n-1)(2^n -1)$ losses or ties overall in phase 1.

    Next, we show we correctly record a best response to $a_{f(n)}$ in phase 2. When $i=n$, the opponent does not necessarily switch as soon as the average payoff of $a_{f(n)}$ becomes negative; at this point, every action has a negative average payoff, and it is possible that the average payoff of $a_{f(n)}$ becomes negative but remains the highest. However, the opponent will eventually switch from $a_{f(n)}$ if we play an action which beats it enough times: there must exists some other action $a$ which has an average payoff strictly higher than $-1$, so as we play an action which beats $a_{f(n)}$ and its payoff approaches $-1$, it will eventually reach a lower average payoff than $a$ and the opponent will switch. Suppose this was not the case and every other action $a_{f(i)}$ for $i < n$ had an average payoff of exactly $-1$. In that case, the opponent must have switched away from every action $a_{f(i)}$ immediately after the first round, when it was beaten by the action we played in that round. We play each action in succession until the opponent switches, not starting over when the opponent switches, so this must mean the first action we play beat every action $a_{f(i)}$ for $i < n$. However, each action must tie against at least one action (itself) and lose to at least one, so it is not possible that a single action beats $n-1$ of the $n$ actions. 

    The lowest possible $\textit{best}$ payoff for any action besides $a_{f(n)}$ would therefore have been achieved by one tie followed by a loss and switch, for an average payoff of $-\frac{1}{2}$. So in order to get the opponent to switch from $a_n$ to the action with the next highest payoff, we need to make its payoff $\leq -\frac{1}{2}$ in the worst case (since $a_{f(n)}$ is the last action in the opponent's action order, its average payoff reaching exactly the value of the next-highest average payoff would cause a switch). In the worst case, if the opponent won every previous round they played $a_{f(n)}$, we would have to play an action which beat it 3 times the number of times the opponent played it so far to reach a payoff of $-\frac{1}{2}$. This means we play the first action 1 time, the second action 3 times, followed by 12, 48, and so on, which is as indicated in the algorithm. This requires at most $4^{n-1}$ rounds total. 

    Because the opponent switches from $a_{f(n)}$ either when its average payoff goes from non-negative to negative or from negative to lower, its net score must have decreased before the switch. Therefore, the round preceding the switch was a loss for the opponent and we correctly record our action in that round as a best response to $a_{f(n)}$. 
    
    During phase 1, we incur at most $n2^n - 2^n - n + 1$ losses or ties. In phase 2, we incur at most $4^{n-1}$ losses or ties. Then, since we correctly recorded best responses, we only tie or lose during phase 3 if the ellipsoid algorithm makes a prediction mistake. The ellipsoid algorithm makes at most $O(n^4 \log(nr))$ mistakes (see $\ref{sec_ellipsoid_mistake_bounds}$ for proof). Therefore, the total number of losses or ties is bounded by $O(4^{n-1} + n^4\log(nr))$.\end{proof}

\end{document}